\pdfoutput=1
\newif\ifFull
\Fulltrue
\ifFull
\documentclass[11pt]{article}
\usepackage[margin=1.25in]{geometry}
\else
\documentclass{llncs}
\renewenvironment{proof}{\noindent{\bf Proof:}}{\hspace*{\fill}\qed\bigskip}
\spnewtheorem{appthm}{Theorem}[section]{\bfseries}{\itshape}
\fi

\usepackage{graphicx}
\usepackage{amsmath}
\usepackage{amssymb}
\usepackage{enumerate}
\usepackage{verbatim}
\usepackage{hyperref} 
\usepackage{cite}

\ifFull
\usepackage{amsthm}
\newtheorem{theorem}{Theorem}[section]
\newtheorem{appthm}{Theorem}[section] 
\newtheorem{lemma}[theorem]{Lemma}

\newtheorem{corollary}[theorem]{Corollary}

\fi

\newcommand{\mydiv}{\textrm{ div }}
\newcommand{\R}{\mathbb{R}}
\newcommand{\N}{\mathbb{N}}

\title{Fully Retroactive Approximate Range \\ 
and Nearest Neighbor Searching}
\ifFull
\author{Michael T. Goodrich \qquad Joseph A. Simons \\
Department of Computer Science, University of California, 
	   Irvine, USA.
}
\else
\author{Michael T. Goodrich \and Joseph A. Simons }
\institute{Department of Computer Science, University of California, 
	   Irvine, USA.}
\fi
\date{}

\begin{document}

\maketitle

\begin{abstract}
We describe fully retroactive dynamic data structures
for approximate range reporting 
and approximate nearest neighbor reporting. 
\ifFull
We show how to maintain, for any positive constant $d$, a set of $n$
points in $\R^d$ indexed by time such that we can perform
insertions or deletions at any point in the timeline in 
$O(\log n)$ amortized time. We support, for any small constant $\epsilon>0$,
$(1+\epsilon)$-approximate range reporting queries at any point in the
timeline in $O(\log n + k)$ time, where $k$ is the output size.
We also show how to answer $(1+\epsilon)$-approximate nearest neighbor
queries for any point in the past or present in $O(\log n)$ time.
\fi
\end{abstract}

\pagestyle{plain}

%
%
\section{Introduction}
Spatiotemporal data 
types are intended to represent objects that have geometric
characteristics that change over 
time\ifFull{ (e.g., see~\cite{hktg-eiso-02,mghsv-stdt-99,tr-sstd-99}).}
\else{.}
\fi

The important feature of such objects is that their critical
characteristics, such as when they appear and disappear 
in a data set, exist in a \emph{timeline}.
The representation of such objects has a number of important
applications, for instance,
in video and audio processing, geographic information
systems, and historical archiving.
Moreover, due to data editing or cleaning needs, 
spatiotemporal data sets may need to be updated in a dynamic fashion,
with changes that are made with respect to the timeline.
\ifFull
Some possible examples of such dynamic updates include:
\begin{itemize}
\item
A video editor may wish to add a 3-second video 
segment at the 20-second mark in a 27-second video and have the result
be a seamless 30-second video (say, for a television commercial).
\item
A credit reporting company might need to remove some historical
transactions in a consumer's credit report to reflect the fact that
these transactions were entered in error.
\item
Some of the historical trajectories in a collection of GPS traces may need to be
changed to correct calibration errors that are discovered only after
subsequent post-processing.
\end{itemize}
\fi
Thus, in this paper we are interested in methods for dynamically maintaining
geometric objects that exist in the context of a timeline.
Queries and updates happen in \emph{real time}, but are indexed in
terms of the timeline. 
\ifFull
For instance,
one can ask to mark an object to exist for the first time 
at time index $t_0$, that is, to be inserted at time $t_0$.
Likewise, one may ask
to mark an object so it is identified as 
removed as of time index $t_1$, that is, to be deleted at time $t_1$.
One may also ask to query the data set as of time index $t_2$, say,
to ask for an approximate nearest neighbor of a point $p$ as of time
index $t_2$.
These updates and queries occur in real time with the intent that the
timeline for the data set is made immediately consistent after each
update and each query is performed correctly with respect to the
current state of the timeline.

For example, suppose we have a set of 1-dimensional points, 
$X=\{1, 4, 7, 10, 13\}$, so that each $x\in X$ is inserted into
the timeline at time index $t=x$ and never removed.
If we then do a nearest-neighbor query for $6$ at time index $12$,
the result would be $7$.
But if we then revise the timeline so that $6$ is inserted at time
index $6$, and we repeat the above nearest-neighbor query, then the
answer would be $6$.
Thus, our intent is that each update we make to the data set should
propagate through the timeline in a consistent fashion, so that future
queries made with respect to the timeline are answered correctly.
Of course, we cannot go back in real time to change the answers to
queries done in the (real) past with respect to the state of the
timeline when that query was made.
\fi

In this paper, we are specifically interested in the dynamic
maintenance of a set of $d$-dimensional points that appear and
disappear from a data set in terms of indices in a timeline, for a
given fixed constant $d\ge 1$.
Points should be allowed to have their appearance and disappearance
times changed, with such changes reflected forward in the timeline.
We also wish to support time-indexed approximate range reporting
and nearest-neighbor queries in such data sets.
That is, we are interested in the dynamic maintenance of spatiotemporal 
point sets with respect to these types of geometric queries.


\subsection{Related Work}

\paragraph{Approximate Searching.}
Arya and Mount~\cite{am-annqf-93} introduce the approximate nearest
neighbor problem for a set of points, $S$, such that given a query point
$q$, a point of $S$ will be reported
whose distance from $q$ is at most a factor
of $(1 + \epsilon)$ from that of the true nearest neighbor of $q$.
Arya {\it et al.}~\cite{amnsw-oaann-98} show that such queries can be
answered in $O(\log n)$ time for a fixed constant $\epsilon>0$.
Chan~\cite{Chan1998} shows how to achieve a similar bound.
Arya and Mount~\cite{am-ars-00} also introduce the approximate range
searching problem for a set, $S$, 
where a axis-parallel rectangle $R$ is given as input
and every point in $S$ that is inside $R$ by a distance of at least
$\epsilon$ is reported as output and no point of $S$ outside $R$ is
reported. Let $k$ be the number of points reported.
Arya and Mount show that such queries can be answered in $O(\log n + k)$
time for fixed constant $\epsilon>0$.
Eppstein {\it et al.}~\cite{egs-sqsdd-05} describe the skip quadtree
structure, which supports $O(\log n+k)$-time approximate range
searching as well as $O(\log n)$-time point insertion and deletion.

Our approach to solving approximate 
range searching and approximate
nearest neighbor problems are based on the quadtree 
structure~\cite{o-mtuas-82}.
In this structure,
regions are defined by squares in the plane, which are
subdivided into four equal-sized squares for any regions
containing more than a single point. So each internal
node in the underlying tree has up to four children and regions have
optimal aspect ratios.
Typically, this structure is organized in a 
\ifFull
compressed fashion~\cite{as-dchan-99,b-acpqh-93,Bern1993,c-faann-83},
\else
compressed fashion~\cite{Bern1993},
\fi
so that 
paths in the tree consisting of nodes with only one non-empty child
are compressed to single edges.  
This structure is related to the
balanced box decomposition (BBD) trees
of Arya {\it et al.}~\cite{am-annqf-93,am-ars-00,amnsw-oaann-98},
where regions are defined by hypercubes
with smaller hypercubes subtracted away, so that the height of the
decomposition tree is $O(\log n)$.  
Similarly, Duncan {\it et al.}~\cite{dgk-bartc-01} define the
balanced aspect-ratio (BAR) trees, where
regions are associated with
convex polytopes of bounded aspect ratio, 
so that the height of the decomposition tree is $O(\log n)$.

\paragraph{Computational Geometry with respect to a Timeline.}
Although we are not familiar with any previous work on 
retroactive $d$-dimensional approximate range searching and nearest-neighbor
searching,
we nevertheless would like to call attention to the fact that
incorporating a time dimension to geometric constructions and data
structures is well-studied in the computational geometry literature.
\begin{itemize}
\item
Atallah~\cite{Atallah19851171} studies several 
\emph{dynamic computational geometry}
problems, including convex hull maintenance, for points moving
according to fixed 
\ifFull
trajectories,
showing an important connection between such problems and Davenport-Schinzel sequences.
This work has been followed by 
a large body of subsequent work 
on additional connections between geometry, 
moving objects, and Davenport-Schinzel sequences.
(E.g., see~\cite{sa-dsstg-95}.)
\else
trajectories.
\fi
\item
Subsequently, a number of researchers have studied geometric motion
problems in the context of \emph{kinetic data structures}
\ifFull
(e.g., see~\cite{aeg-kbisd-98,bgsz-pekds-97,g-kdssar-98,ghsz-kcud-00}).
\else
(e.g., see~\cite{g-kdssar-98}).
\fi
In this framework, a collection of geometric objects is moving
according to a fixed set of known trajectories, and changes can only
happen in the present.
\ifFull
The goal is to
maintain a data structure that supports geometric queries on this set
with respect to the ``current'' time. As time progresses, the data
structure needs to be updated, either because internal conditions
about its state are triggered or because an object changes its
trajectory.
\fi
\item
\ifFull
Sarnak and Tarjan~\cite{st-pplup-86} and
\fi
Driscoll {\it et al.}~\cite{dsst-mdsp-89}
introduce the concept of \emph{persistent data structures}, which
support time-related operations where updates occur in the 
present and queries can be performed in the past, but updates in the past
fork off new timelines rather than propogate changes forward in the same
timeline.
\ifFull
Such structures have been used in a number of applications, such as
in planar point location, which use
space-sweeping operations to construct data structures based on a
static-to-dynamic-to-static framework.
\fi
\end{itemize}
All of this
previous work 
differs from the
approach we are taking in this paper, since
in these previous approaches objects
are not expected to be retroactively changed ``in the past.''

Demaine {\it et al.}~\cite{Demaine2007} introduce the concept of 
\emph{retroactive data structures}, which is the framework we follow
in this paper.
In this approach, a set of data is maintained with respect to a
timeline. Insertions and deletions are defined with respect to this
timeline, so that each insertion has a time parameter, $t$, and so
does each deletion.
Likewise, queries are performed with respect to the time parameter as
well.
The difference between this framework and the dynamic computational
geometry approaches mentioned above, however, is that updates can be done
retroactively ``in the past,'' with the changes necessarily being
propagated forward.
If queries are only allowed in the current state (i.e., with the
highest current time parameter), then the data structure is said to
be \emph{partially retroactive}. If queries can be done at any point
in the timeline, then the structure is said to be \emph{fully
retroactive}.
Demaine {\it et al.}~\cite{Demaine2007} describe a number of results
in this framework, including a fully-retroactive 1-dimensional 
structure for successor queries with $O(\log^2 n)$-time performance.
They also show that any data structure for a decomposable search
problem can be converted into a fully retroactive structure at a cost
of increasing its space and time by a logarithmic factor.

Acar {\it et al.}~\cite{Acar2007} introduce an alternate model of retroactivity,
which they call \emph{non-oblivious} retroactivity.
In this model, one maintains the historical sequence of queries as well 
as insertions and deletions.
When an update is made in the past, 
the change is not necessarily propagated all the way forward to the present.
Instead, a non-oblivious data structure returns the first operation 
in the timeline that has become \emph{inconsistent},
that is an operation whose return value has changed because of the 
retroactive update. 
As mentioned above,
we only 
consider the original model of retroactivity as defined 
by Demaine {\it et al.}~\cite{Demaine2007} in this paper.

Blelloch~\cite{Blelloch2008} and Giora and Kaplan~\cite{Giora2009} 
consider the problem of maintaining a fully retroactive 
dictionary that supports successor or predecessor queries. 
They both base their data structures on a structure by 
Mortensen~\cite{Mortensen2003}, which answers fully retroactive one 
dimensional range reporting queries, although Mortensen framed the problem 
in terms of two dimensional orthogonal line segment intersection reporting. 
In this application, the $x$-axis is viewed as a timeline for a retroactive
data structure for 1-dimensional points.
The insertion of a segment $[(x_1,y),(x_2,y)]$ corresponds to the
addition of an insert of $y$ at time $x_1$ and a deletion of $y$ at
time $x_2$.
Likewise, the removal of such a segment corresponds to the removal of
these two operations from the timeline.
For this 1-dimensional retroactive data structuring problem,
Blelloch and
Giora and Kaplan give data structures that support
queries and updates in $O(\log n)$ time.
Dickerson {\it et al.}~\cite{dickerson:cloning} describe a
retroactive data structure for maintaining the lower envelope of a
set of parabolic arcs and give an application of this structure to
the problem of cloning a Voronoi diagram from a server that answers
nearest-neighbor queries.

\subsection{Our Results}
In this paper, we describe fully retroactive dynamic data structures for
approximate range reporting and approximate nearest neighbor searching. We show
how to maintain, for any positive constant, $d\ge 1$, a set of $n$ points in
$\R^d$ indexed by time such that we can perform insertions or deletions at any
point in the timeline in $O(\log n)$ amortized time. We support, for any small
constant $\epsilon>0$, $(1+\epsilon)$-approximate range reporting queries at any
point in the timeline in $O(\log n + k)$ time, where $k$ is the output size.
Note that in this paper we consider circular ranges defined by a query point $q$
and radius $r$.  We also show how to answer $(1+\epsilon)$-approximate nearest
neighbor queries for any point in the past or present in $O(\log n)$ time.  Our
model of computation is the real RAM, as is common in computational geometry
algorithms (e.g., see~\cite{ps-cgi-90}). 

The main technique that allows us to achieve these results
is a novel, multidimensional version of
fractional cascading, which may be of independent interest.
Recall that in the (1-dimensional)
\emph{fractional cascading} paradigm of Chazelle and
Guibas~\cite{cg-fc1ds-86,cg-fc2a-86}, 
one searches a collection of sorted lists (of what are essentially numbers),
which are called \emph{catalogs},
that are stored along nodes in a search path of a catalog graph, $G$,
for the same element, $x$.
In \emph{multidimensional fractional cascading}, one instead searches 
a collection of finite subsets of $\R^d$
for the same point, $p$, along nodes in a search path of 
a catalog graph, $G$.
In our case, rather than have each catalog represented 
as a one-dimensional sorted list, we instead represent each 
catalog as a multidimensional ``sorted list,'' with points ordered as they would
be visited in a space-filling curve (which is strongly related to how the
points would be organized in a quadtree, e.g., 
\ifFull
see~\cite{Samet1984,Bern1993}).
\else
see~\cite{Bern1993}).
\fi{}

By then sampling in a fashion inspired by one-dimensional fractional cascading,
we show\footnote{The details for our constructions 
	are admittedly intricate, so some details
	of proofs are given in appendices.}
how to efficiently perform
repeated searching of multidimensional catalogs stored at the nodes of a 
search path in a suitable catalog graph, such as a
segment tree (e.g., see~\cite{ps-cgi-90}), with each of the searches
involving the same $d$-dimensional point or region. 

Although it is well known that space-filling curves can be applied to the
problem of approximate nearest neighbor searching, we are not aware of any
extension of space-filling curves to approximate range reporting. Furthermore,
we believe that we are the first to leverage space-filling curves in order to
extend dynamic fractional cascading into a multi-dimensional problem.

\section{A General Approach to Retroactivity}
\label{sec:strategy}
Recall that a query $Q$ is 
\emph{decomposable} if there is a binary operator 
$\square$ (computable in constant time) such that 
$Q(A\cup B) = \square(Q(A), Q(B))$. 
%
%
%
Demaine {\it et al.}~\cite{Demaine2007} showed that we can make any
decomposable search problem retroactive by maintaining each element ever
inserted in the structure as a line segment along a time dimension between the
element's insertion and deletion times.  Thus each point $p$ in $R^d$ 
is now represented by a line segment parallel to the time-axis in $R^{d+1}$
dimensions.  For example when extending the query to be fully-retroactive, a
one-dimensional successor query becomes a two-dimensional vertical ray shooting
query, and a one-dimensional range reporting query becomes a two-dimensional
orthogonal segment intersection 
\ifFull
query 
(Figure \ref{fig:retroSuccessor}).
\begin{figure}
\centering
\includegraphics[width=.95\textwidth]{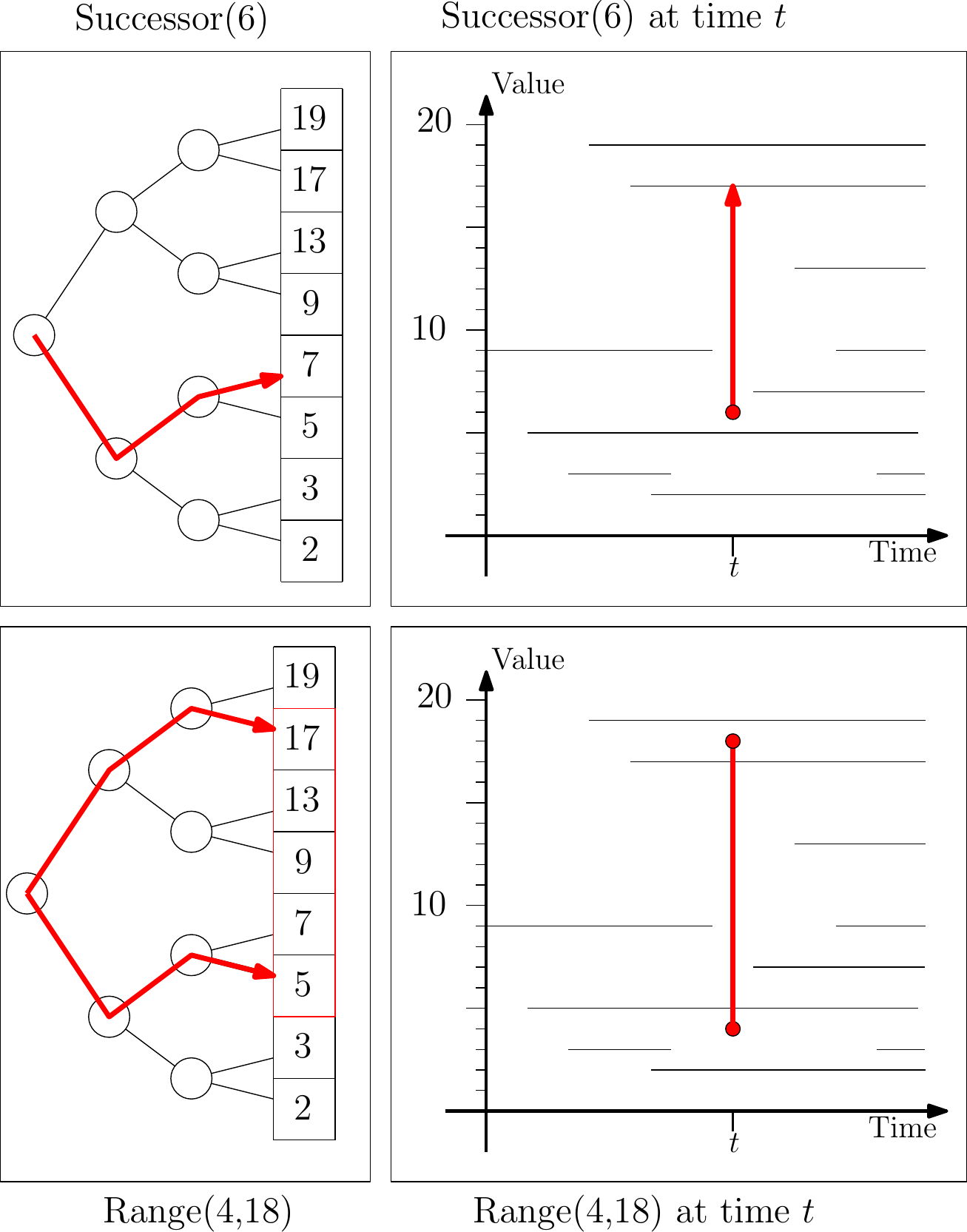}
\caption{
\label{fig:retroSuccessor}
Examples of how a $d$ dimensional query on points becomes a $d+1$
dimensional query on segments. Here we illustrate how a one-dimensional
successor query becomes a two dimensional vertical ray shooting query, and a 
one-dimensional range reporting query becomes a two-dimensional segment intersection
query.}
\end{figure}
\else
query.
\fi

Thus, we maintain a segment tree to allow searching over the segments in the
time dimension, and augment each node of the segment tree with a secondary
structure supporting our original query in $d$ dimensions. Let $S$ be the set of
nodes in the segment tree on a root-to-leaf
path searching down for $t$ in the time dimension.  To answer a
fully-retroactive query,  we perform the same $d$-dimensional query at each node
in $S$.  This transformation costs an extra $\log n$ factor
in space, query time, and update time, which we would nevertheless like to
avoid. 

Recall that Mortensen \cite{Mortensen2003} and Giora and Kaplan \cite{Giora2009}
both solve the fully-retroactive versions of decomposable search
problems, and are both able to avoid the extra $\log n$ factor in query and update
time.  Therefore inspired by their techniques, we propose the following as a
general strategy for optimally solving the 
fully-retroactive version of any decomposable
search problem.  
\begin{enumerate}[1.]
\item Suppose we have an optimal data structure $D$ for the non-retroactive problem which supports
queries in polylogarithmic time $T(n)$.
\item Represent each $d$-dimensional point as a line segment in $d+1$
dimensions.
\item Build a weight-balanced segment tree with polylogarithmic branching factor
over the segments as described by \cite{Giora2009}.
\item Augment the root of the segment tree with an optimal search structure
$D$.
\item Augment each node of the segment tree with a colored dynamic fractional cascading
(CDFC) data structure.    
\item Perform a retroactive query at time $t$ 
by performing the non-retroactive query on the non-retroactive data structure at
the root of the segment tree, and for each node on the search path for $t$ in
the segment tree, perform the query in each of the CDFC structures 
(Figure \ref{fig:AugSegTree}).
\end{enumerate}

The CDFC data structure must be cleverly tuned to support a \emph{colored} 
(but non-retroactive) version of the $d$-dimensional query in $O(T(n) \cdot \log
\log n /\log n)$
time. In a colored query, each element in the structure is given a color, and
the query also specifies a set of colors. We only return elements whose color is
contained in the query set.
The colors are required because the segment tree has high degree. Each color
represents a pair of children of the current node in the segment tree. Thus we
encode which segments overlap the search path via their colors.
Since the segment tree has a polylogarithmic branching factor, we spend $T(n)$
time searching at the root and $O(T(n) \cdot \log \log n / \log n)$ time
searching in the CDFC structures at each of the $O(\log \log n / \log n)$ nodes.
Therefore, the total time required by a query is an optimal $O(T(n))$.
Updates follow a similar strategy, 
but may require us to periodically rebuild sections of the segment tree.
We can still achieve the desired (amortized) update time, and the analysis
closely follows \cite{Giora2009}.
\begin{figure}[hbt!]
\centering
\includegraphics[width=.8\textwidth]{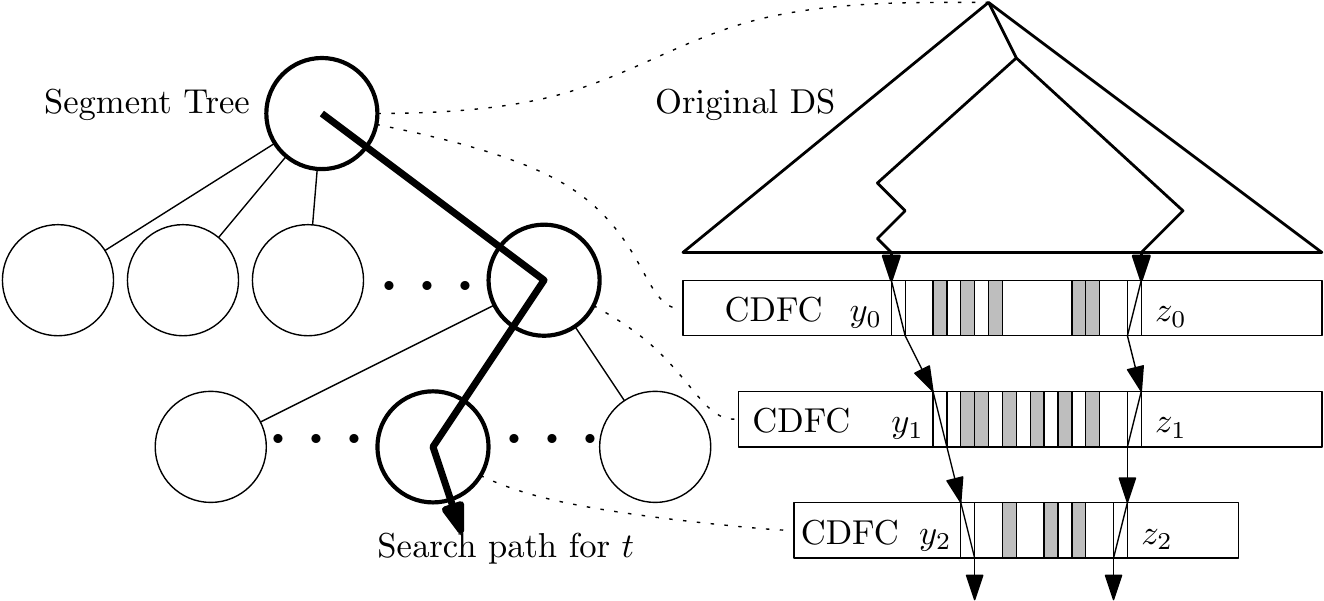}
\caption{
    \label{fig:AugSegTree}
   An application of our strategy to a fully-retroactive one-dimensional range
reporting query. Shaded elements represent elements with colors indicating they are present at time $t$.
} 
\end{figure}

One of the key difficulties in applying this strategy lies in the
design of the colored dynamic fractional cascading data structure, especially
in problems where the dimension $d > 1$. 
In fact, the authors are not aware of any previous application 
of dynamic fractional cascading techniques to any multidimensional search problem.
However, in the following we show how techniques using space filling curves can
be applied to extend the savings of fractional cascading into a multidimensional
domain. First, we apply the above strategy in the simpler case when
$d = 1$. Then we extend this result using space filling curves to support
Fully-Retroactive range reporting queries and approximate nearest neighbor
queries in $\R^d$.
Note that in one dimension a nearest
neighbor query reduces to a successor and predecessor query.
\begin{lemma}
\label{lem:CDFC}
There exists a colored dynamic fractional cascading data structure which
supports updates in $O(\log \log N)$ amortized time, 
colored successor and predecessor queries in $O(\log \log N)$ worst case time and 
colored range reporting queries in $O(\log \log N + k)$ worst case time, where
$N$ is the number of elements stored and $k$ is the number of elements reported.
\end{lemma}
\begin{proof}
We extend the generalized union-split-find structure of 
\cite{Giora2009} to also support 
colored range queries. See Appendix~\ref{app:gusf}.
\end{proof}
\paragraph{Space Filling Curves.}
The z-order, due to Morton~\cite{Morton1966}, is commonly used to map multidimensional points down to one dimension.
 By now space filling curves are well-studied
and multiple authors have applied them specifically to 
 \ifFull
quadtrees (e.g., see~\cite{Samet1984,Bern1993}) and
\else
\fi{}
approximate nearest neighbor queries~\cite{Derryberry2008,Chan2006,Chan1998,Liao2001}.
However, we extend their application to approximate range searching as well. 
Furthermore, we believe that we are the first to leverage space-filling curves to extend dynamic fractional
cascading techniques to multidimensional problems such as these.
\begin{lemma}
\label{shuffleOrder}
The set of points in any quadtree cell rooted at $[0, 1)^d$ form
a contiguous interval in the z-order.
\end{lemma}
\begin{proof}
Due to Bern {\it et al.}~\cite{Bern1993}. See Appendix~\ref{app:space} 
for details.
\end{proof}
\begin{lemma}
\label{cANNthm}
Let $P$ be a set of points in $\R^d$. 
Define a constant $c = \sqrt{d}(4d+4) + 1$. 
Suppose that we have $d+1$ lists $P + v^{j}, j= 0, \ldots, d$, each one sorted according to its z-order.  We can find a query point $q$'s $c$-approximate nearest neighbor in $P$ by examining the the $2(d+1)$ predecessors and successors of $q$ in the lists. 
\end{lemma}
\begin{proof}
See Appendix~\ref{app:space}.
\end{proof}
\section{Main Results}
In this section we give our primary results: data structures for
fully-retroactive approximate range queries and fully-retroactive approximate
nearest neighbor (ANN) queries.
Recall that an approximate range query $report(q,r,\epsilon, t)$ defines an
inner range $Q^-$, the region within a radius $r$ of the query point $q$  and an
outer range $Q^+$, the area within a radius of $(1+\epsilon)r$ of $q$. We  want
to return all the points inside $Q^-$ and exclude all points outside $Q^+$ for a
particular point in time $t$. Points that fall between $Q^-$ and $Q^+$ at time
$t$ may or may not be reported. Points not present at time $t$ will never be
reported.
\begin{theorem}[Fully-Retroactive Approximate Range Queries]
For any positive constant $d\ge 1$, we can maintain  a set of $n$ points in
$\R^d$ indexed by time such that we can perform insertions or deletions at any
point in the timeline in $O(\log n)$ amortized time. We support 
for any small
constant $\epsilon>0$, $(1+\epsilon)$-approximate range reporting queries at any
point in the timeline in $O(\log n + k)$ time, where $k$ is the output size.
The space required by our data structure is $O(n \log n /
\log \log n)$. 
\end{theorem}
\begin{proof} 
We follow the general strategy outlined in Section
\ref{sec:strategy}. We augment the root of the segment tree with a
skip-quadtree, an optimal structure for approximate range and nearest
neighbor queries in $\R^d$. We also augment each node of the segment tree with a
specialized CDFC structure which we now describe. 

We extend the CDFC structure from Lemma \ref{lem:CDFC} to maintain
$d$-dimensional points such that given a query set of colors $C_q$ and  
$d$-dimensional quadtree cell, it returns all points contained in that
cell with colors in $C_q$.
By Lemma \ref{shuffleOrder}, for all points $y,q,z$ such that $y < q < z$ in the z-order, any
quadtree cell containing $y$ and $z$ must also contain $q$. Furthermore, for a given
quadtree, each cell is uniquely defined by the leftmost and rightmost leaves in
the cell's subtree. Therefore, the $d$-dimensional cell query reduces to a
one-dimensional range query in the z-order on the unique points which define the
quadtree cell (Figure \ref{fig:zorder}). Thus, by leveraging Lemma
\ref{lem:CDFC} and maintaining the points according to their z-order, we support
the required query in $O(\log \log n + k)$ time.  

 
The skip-quadtree tree contains all the points ever
inserted into our data structure, irrespective of the time that they were
inserted or deleted. The inner and outer range of a query partition the set of
quadtree cells into the three subsets: the \textit{inner} cells, which are
completely contained in $Q^+$, the \textit{outer cells}, which are completely
disjoint from $Q^-$, and the stabbing cells, which partially overlap $Q^+$ and
$Q^-$.  Eppstein {\it et al.}~\cite{Eppstein2005} showed that a skip-quadtree can
perform an approximate range query in $O(\log n + \epsilon^{1-d})$ time expected
and with high probability, or worst case time for a deterministic skip-quadtree.
Using their algorithm we can find the set $I$ of $O(\epsilon^{1-d})$ disjoint
inner cells in the same time bound. Based on the correctness of their algorithm
we know that the set of points we need to report are covered by these cells. We
report the points present at time $t$ as follows: For each cell $i \in I$, if
$i$ is a leaf, we report only the points in $i$,
which are present at time $t$ in
constant time. Otherwise, we find the first and last leaves $y_0$ and
$z_0$ according to an in-order traversal of $i$'s subtree in $O(\log n)$
time. Then we perform a fully-retroactive range query on cell $i$ in the segment
tree, using $z_0$ and $y_0$ to guide the query on cell $i$ in the CDFC structures. 
Correctness
follows since a point satisfies the retroactive range query if and only if it is in one
of the cells we examine and is present at time $t$.   

For each of the $O(\epsilon^{1-d})$ cells in $I$ we spend $O(\log n +
k_i)$ time and so the total running time is $O(\epsilon^{1-d} \log n + k) =
O(\log n + k)$, since $\epsilon$ is a small constant. The skip-quadtree is
linear space, and thus the space usage is dominated by the segment tree.
The total space required is $O(n \log n / \log \log n)$.
\end{proof}

\begin{corollary}[Fully-Retroactive ANN Queries.]
We can maintain, 
for any positive constant, $d\ge 1$, a set of $n$ points in
$\R^d$ indexed by time such that we can perform insertions or deletions at any
point in the timeline in $O(\log n)$ amortized time. We support, 
for any small
constant $\epsilon>0$, $(1+\epsilon)$-approximate nearest
neighbor queries for any point in the past or present in $O(\log n)$ time.
The space required by our data structure is $O(n \log n /
\log \log n)$. 
\end{corollary}
\begin{proof}
By combining the data structure of \cite{Giora2009} with Lemma \ref{cANNthm}
and storing the $d$-dimensional points in that structure according to their
z-order, we already have a data structure for fully-retroactive $c$-approximate
nearest neighbor queries. However, $c$ is polynomial function of $d$, and for $d
= 2$, $c$ is already greater than 15.  In order to support $(1+ \epsilon)$
nearest neighbor queries for any $\epsilon > 0$, we require the data structure
of the previous theorem. 

We know from Lemma \ref{cANNthm} that we can find a
$c$-approximate nearest neighbor by using $d+1$ different shifts of the points.
Therefore, we augment our data structure so that instead of a single CDFC
structure at each segment tree node, 
we keep an array of $d+1$
CDFC structures corresponding to the z-order of each of the $d+1$ sets of
shifted points. Given a point $q$, we can find a the predecessor and successor
of $q$ at time $t$ in each of the $d+1$ z-orders in $O(d \log n)$ time. Out of
these $2(d + 1)$ points, let $p$ be the point that is closest to $q$. By
Lemma \ref{cANNthm}, $p$ is a $c$-approximate nearest neighbor. Let $r$ be the
distance between $p$ and $q$. As observed by multiple authors
\cite{Arya2010,Mount2010,Har-Peled2001}, we can find a
$(1+\epsilon)$-approximate nearest neighbor via a bisecting search over the
interval $[r/c, r]$. This search requires $O(\log(1/ \epsilon))$
fully-retroactive spherical emptiness queries. 
We can support a retroactive spherical emptiness query in
$O(\log n)$ time with only a slight modification to our retroactive approximate range query.
Instead of returning $k$ points in the range, we just return the first point we
find, or \texttt{null} if we find none. Thus the total time required is $O(d\log
n + \log(1/\epsilon)\log n) = O(\log n)$ since we assume $\epsilon$ and $d$ are
constant. The space usage only increases by a factor of $d$ when we store the
$d+1$ shifted lists, and thus the space required is still $O(n \log n /
\log \log n)$. 

\end{proof}



\ifFull
\subsection*{Acknowledgments}
This research was supported in part by the National Science
Foundation under grants 0847968 and 0953071 and by the Office
of Naval Research under (MURI) Award number N00014-08-1-1015.
\fi

\bibliographystyle{acm}
\bibliography{geom,goodrich,mybib}

\clearpage
\begin{appendix}
\section*{Appendices}

\section{Generalized Van Emde Boas (GVEB) Trees}
\label{app:gveb}

\begin{appthm}
A GVEB structure with parameters $(N,M)$ uses $O(NC)$ space and can be
initialized in $O(NC)$ worst case time, where $C = \log^\frac{1}{4}M$ is the
number of colors supported. It
requires a lookup table of size $O(M)$ and supports \texttt{insert},
\texttt{delete} and
colored successor (\texttt{find}) queries in $O(\log \log N)$ worst case time. It also
supports colored range reporting queries (\texttt{report}) in $O(\log \log N + k)$ time,
where $k$ is the number of elements reported. 
\end{appthm}

The Generalized van Emde Boas tree (GVEB) is an extension of the data structure
by van Emde Boas {\it et al.}~\cite{VanEmdeBoas1976}.  Recall that a van Emde Boas
tree (VEB) is a recursive data structure that supports successor and predecessor
queries for a fixed (monochrome) universe $U$ in $O(\log \log U)$ time. A
recursive data structure is one that references itself over a smaller section of
input as a part of its definition. For example, in a binary search tree, an
internal node $v$ with key $k$ partitions the set of input keys into keys
greater than $k$ that are stored recursively in the right subtree and keys less
than $k$ that are stored recursively in the left subtree. Thus the children of
each internal node are themselves binary search trees defined over a subset of
the input. 

To avoid tedious details,
we make the typical assumption that $U$ is the set of integers in the range
$[N]$ and that $N$ has the form $2^{2^\ell}$ for some positive integer $\ell$.
Note that $\ell = O(\log \log N)$ is the number of levels of recursive
structures.  
Each node of the VEB tree contains the following four fields: \texttt{Min},
\texttt{Max}, \texttt{Bottom}, and \texttt{Top}. The \texttt{Min} and
\texttt{Max} fields contain the minimum and maximum elements in the node's
subtree. 
The VEB partitions $[N]$ into $\sqrt{N}$ buckets each of size $\sqrt{N}$. These
buckets are stored in \texttt{Top}, a single recursive structure. Each bucket in
\texttt{Top} points to a recursive VEB structure for a universe size of
$\sqrt{N}$ in \texttt{Bottom}. The recursive structures in \texttt{Bottom} can
be thought of as the children of the node, and the single structure in
\texttt{Top} functions as an efficient way to search these children. 
If a key $k$ is the maximum or minimum integer in a particular recursive
structure,
then it is just stored in the Max or Min fields. Otherwise, it is stored
recursively in \texttt{Bottom}$[k \mydiv \sqrt{N}]$ as $k \mod \sqrt{N}$.

Let $N,M$ be two integers that fit in a machine word, and let $C = [\log
^{\frac{1}{4} }M]$. 
The Generalized van Emde Boas (GVEB) data structure of
\cite{Giora2009} maintains a set of (key, color) pairs $(k,c)$ where $k$ is an
integer in $[N]$ and $c$ is an integer in $C$.  Let $C_q \subseteq C$ be a set
of colors. The query \texttt{find}$(k_q,C_q)$ returns the successor of $k_q$ with
color $c \in C_q$, i.e. we want to return the minimum key element $(k', c')$ 
 such that $k' \geq k_q$ and  $c' \in C_q$. 
GVEB supports find, insert and delete operations in
$O(\log\log N)$ worst case time.

The GVEB tree is similar to the VEB tree, only instead of maintaining a single
minimum and maximum integer in each recursive structure, we
maintain a min and max for each color $c \in C$. Let $G$ be a GVEB structure.
The \texttt{Min} and \texttt{Max} arrays are actually maintained via two arrays
each (4 total) of size $|C|$: \texttt{max-key, max-rank, min-key, min-rank}. For
each $c \in C$, \texttt{max-key}$[c] = \max \{k\ | (k,c) \in G \}$ 
and
\texttt{max-rank}$[c]$ is the rank of \texttt{max-key}$[c]$. The arrays
\texttt{min-key} and \texttt{min-rank} are defined symmetrically. \texttt{Max}
and \texttt{Min} also each contain a \emph{q-heap}, max-Q and min-Q, which
contain all the elements in Max and Min respectively.  Recall that the \emph{q-heap} of
\cite{Fredman1994} can store up to $\log^\frac{1}{4}M = |C|$ elements, requires
a lookup table of size $M$,  and supports successor and rank queries in constant
time.

Finally, we precompute a table of size $O(M)$ that points to the maximum
element in $G$ with a color in $C_q$ for every possible set $C_q$ and for each
possible \texttt{max-rank} array.  Together these data structures allow us to
spend only constant time at each level of the structure during a find, insert,
or delete operation, so that the time required for each operation is $O(\log
\log n)$. For full details of how these operations work, see section 5.1 of
\cite{Giora2009}. 

In addition to the insert, delete, and find operations, 
we also require a one-dimensional
colored range reporting query. Therefore we extend the GVEB $G$ to also support the
query \texttt{report}$(i,j,C_q)$, which returns the set of elements $\{(k,c) \in
G| i \leq k \leq j, c \in C_q\}$. We also need an internal helper query
\texttt{reportany}$(i,j,C_q)$, which, for each color $c \in C_q$ such that the
set $\{(k,c) | (k,c) \in G, i \leq k \leq j \}$ is non-empty, returns exactly
one element from that set. For each color $c$ such that there exists an element
$(k,c) \in G$, we maintain pointers to the successor and predecessor of $k$ with
color $c$. We can maintain these pointers in $G$ with no additional asymptotic
cost. Then to answer a \texttt{report} query, we first perform a
\texttt{reportany} query, and then follow these pointers to find the rest of the
elements in the range.  We now describe how to perform the \texttt{reportany}
query efficiently. 

We will require the data structure from Lemma 3.2 of~\cite{Mortensen2003}. 
\begin{lemma}
\label{reportLemma}
We can maintain an array $A$ indexed by $C$ where each entry in $A$ is an
integer in $[O(N)]$ such that updates to $A$ can be performed in constant time
and such that given i and j we can calculate the set $\textrm{report}(A,i,j) =
\{ c \in C | i \leq A[c] \leq j \}$ in constant time. The space usage of $A$ is
$O(|C|)$ and $A$ can be initialized in constant time. 
\end{lemma}
\begin{proof}
The structure uses a q-heap and a global lookup table of size $O(N)$ to
maintain in a single word a table $B$ indexed by elements in $C$ such that
$B[C]$ is the rank of $A[c]$ among the elements of $A$.
See~\cite{Mortensen2003} for comparable details.
This structure is
similar to the \texttt{max-rank} data structure of the GVEB structure above, but
instead of returning the rank of a single element, it returns the colors of the
elements between two ranks as a single word. 
\end{proof}

Let $min'$ and $max'$ each be data structures of the type in Lemma
\ref{reportLemma}. We add these structures to each recursive GVEB structure
$G$ to answer the query reportany$(G,i,j,C_q)$. The following algorithm 
follows how Mortensen answered the \texttt{reportany} query in his $S_n$ data
structure~\cite{Mortensen2003}, but differs in some important
details. 
If $G = \emptyset$ or $i >j$, then
there are no elements to report. Otherwise, we find elements to report from the
\texttt{min-key} array as follows. Let $C_r = \texttt{report}(i,j, min')$, and
let $C_{min} = C_q \cap C_r$.   We repeat this process to find elements in the
\texttt{max-key} array. Let $C_q = C_q \backslash C_r$, the subset of query
colors that we have not already reported, and let 
$C_r = \texttt{report}(i,j, max')$, and 
$C_{max} = C_q \cap C_r$. We report the elements min-key$[c]$ for
each $c \in C_{min}$ and max-key$[c]$ for each $c \in C_{max}$. Thus we have
found all colors for which the minimum or maximum element of that color was in
the range $[i,j]$, but there could be other colors for which there exists an
element in the range $[i,j]$ but the minimum and maximum elements of that color
were both outside the range. Therefore, we recursively search for elements with
colors in $C_q = C_q \backslash (C_{min} \cup C_{max})$ via the following
procedure.  If $i = 0$ or $j = N-1$, we return because there can be no
additional elements to report--the max and min elements could not have both been
outside the range. Otherwise, if $i \mydiv \sqrt{N} = j \mydiv \sqrt{N}$, then
the range falls entirely inside a single bucket, and we can find the remaining
elements with a single recursive call: \texttt{reportany}$(G.\texttt{Bottom}[i
\mydiv \sqrt{N}], i \mod \sqrt{N}, j \mod \sqrt{N},C_q)$. Otherwise, the range
spans multiple buckets, and we require 3 recursive calls: 
\begin{enumerate}
\item
\texttt{reportany}$(G.\texttt{Bottom}[i \mydiv \sqrt{N}], i \mod \sqrt{N},
\sqrt{N}-1, C_q)$, 
\item
\texttt{reportany}$(G.\texttt{Top}, i \mydiv \sqrt{N} +
1, j \mydiv \sqrt{N} - 1, C_q)$,
\item
\texttt{reportany}$(G.\texttt{Bottom}[j
\mydiv \sqrt{N}], 0, j \mod \sqrt{N}, C_q)$, passing the $C_q$ parameter by
reference. 
\end{enumerate}
The first recursive call (1) covers the last part of the first
bucket, and the third recursive call (3) covers the first part of the last
bucket. In both (1) and (3), we encounter the stopping condition: $i = 0$ or $j
= N - 1$, and so we only need to examine min-key, max-key, $min'$ and $max'$. It
is only the second recursive call (2) that may require additional recursion.
Note that some of the elements returned by (2) not be actual elements, but
rather pointers to recursive structures. So, as a final step we iterate over all
the elements returned, and while there exists an element of color $c$ that is
actually a pointer to some recursive structure $T$, we replace it with
$T.\texttt{min-key}[c]$. 

We now analyze the additional time and space cost of supporting the
\texttt{report} query. First we note that the data structure of Lemma
\ref{reportLemma} requires no more space than the \texttt{Min} and \texttt{Max}
data structures. Therefore, the asymptotic space usage of our augmented GVEB is
the same as the original GVEB, which was shown by~\cite{Giora2009} to be
$O(N\cdot|C|)$. Second, observe that the data structure of Lemma \ref{reportLemma}
can be updated in constant time, and only needs to be updated when we are
already making changes on a given level of recursion. Therefore updates take
$O(\log \log N)$, no more time than in the original GVEB.  Finally, we consider
the time required to perform a report query. 

The first observation is that whenever we perform more than one recursive call
in the \texttt{reportany} query, only one call requires non-constant work. It's
possible that in the replacement step, $T.\texttt{min-key}[c]$ may point to
another substructure $T'$, and we may need to follow $O(\log \log N)$ pointers
to finally reach a true element. However, the second observation is that for
each substructure $T$ that we examine, there are at least two elements that
fall in the range and which we will ultimately report. Therefore we can answer
the report query in time $O(\log \log n + k)$.

\clearpage
\section{Generalized Union-Split-Find (GUSF)}
\label{app:gusf}
The generalized Union-Split-Find (GUSF) structure is an extension of the
dynamic Union-Split-Find (USF) data structure~\cite{Dietz1991,Raman1993}.
The USF structure maintains a list $L$ of elements from an ordered set, some of
which are marked. It supports the following operations: 
\begin{itemize}
\item 
\texttt{FindNext}$(x)$ returns the smallest element $x' \geq x$ such that
$x'$ is marked. 

\item
\texttt{Add}$(x,y)$ inserts a new unmarked element $x$ just before the element
$y$.

\item
\texttt{Remove}$(x)$ removes an unmarked element $x$ from the list.

\item
\texttt{Unmark}$(x)$ (Union) changes $x$ from a marked element to an unmarked
element.

\item
\texttt{Mark}$(x)$ (Split) changes $x$ from an unmarked element to a marked
element.
\end{itemize}
The USF structure can be considered monochrome. Each element is either unmarked
(white) or marked with a single color (black). The GUSF extends the USF to
support a whole set of colors $C$ so that each element can either be unmarked,
or marked with some subset of colors $C' \subseteq C$. The query and update
operations are extended accordingly. The GUSF also supports additional queries: 
a \texttt{FindPrev} query, which is symmetric to \texttt{FindNext} and a  
\texttt{Report}$(y_1, y_2, C_q)$ query, which 
reports each element $y$ in the list such that $y_1 \leq y \leq y_2$ once for
each color in $C(y) \cap C_q$,
where $C_q$ is a set of query colors and 
$C(y)$ is the set of colors with which element $y$ is marked. 

\begin{appthm}
A GUSF data structure with parameter $N$ uses $O(N)$ space
and supports 
updates in $O(\log \log N)$ amortized time, find queries in  
$O(\log \log N)$ worst case time and report queries in
$O(\log \log N + k)$ worst case time.
The find queries include \texttt{FindNext}$(x, C_q)$ and  
\texttt{FindPrev}$(x,C_q)$ query, which respectively find the successor and
predecessor of $x$ with color $c \in C_q$.
The query \texttt{Report}$(y_1, y_2, C_q)$ reports each element $y$ such that $y_1 \leq y \leq y_2$ once for each color in $C(y) \cap C_q$. 
Supported updates include \texttt{Mark}$(x,c)$, \texttt{Unmark}$(x,c)$,
\texttt{Add}$(x)$ and \texttt{Remove}$(x)$.
\end{appthm}
\begin{proof}

We employ a technique utilized in the Union-Split-Find data structure ~\cite{Dietz1991,Raman1993}
and later adopted by~\cite{Mortensen2003} in his data structure for a colored linked list and~\cite{Giora2009} in their Generalized Union-Split-Find data structure. This technique simultaneously removes the requirement that the elements have integer keys and reduces the space required by our data structure. 
Let $L$ be our list and let $n = |L| \leq N$ be the number of elements in the list.  The key idea is to group consecutive elements of the linked list into blocks of size $\Theta(\log^\beta n)$ for an appropriate small constant $\beta > 3$. We label the blocks with integer labels such that the order of the labels indicates the order of the blocks in $L$. We maintain these labels using the order maintenance structure of 
Willard~\cite{Willard1992}.

In each block $b$ we build a binary search tree $T_b$ and store the elements of $b$ in the leaves of $T_b$, ordered by their order in $L$. We augment the nodes of $T_b$ according to the techniques of~\cite{Mortensen2003} and~\cite{Demaine2007}. Each element $y \in L$ is associated with a set of colors $C(y)$, where $|C(y)| = O(1)$. 
Thus, for each leaf $v_y$ in $T_b$ representing an element $y \in L$, we let $C(v_y) = C(y)$.  
For each internal node $v$ in $T_b$ with children $u$ and $w$, we let $C(v) = C(u) \cup C(w)$. Hence we color each internal node with the union of the colors of its descendant leaves, and we color the root of $T_b$ with the union of the colors of all the elements in block $b$. We say a block has color $c$ if it contains a leaf with color $c$. Thus the colors of the block $b$ equal the colors of the root $r_b$. For each block we maintain a single array $b.\texttt{allleafs}$ of size $|b| = O(\log^\beta n)$ such that for each leaf $v_y \in b$, there is a pointer in $b.\texttt{allleafs}$ to $v_y$. We do not necessarily maintain the order of $b.\texttt{allleafs}$. to match the order of leaves in $b$, but we do maintain a list $b.\texttt{freelist}$ of unused indices in $b.\texttt{allleafs}$. We also keep an array $v.\texttt{leafs}$ indexed by $C$ in each internal node $v$ such that $v.\texttt{leafs}[c] \neq 0$ iff $v$ has a descendant leaf with color $c$, and $b.\texttt{allleafs}[v.\texttt{leafs}[c]]$ is a pointer to one such leaf. 
Each of the indices in $v.\texttt{leafs}$ is $O(\log \log n)$ bits and therefore the entire array $v.\texttt{leafs}$ can be stored in a single word. Furthermore, $C(v)$ can be determined in constant time by examining $v.\texttt{leafs}$. Lastly, for each leaf $v_y$, and color $c \in C(v_y)$, we keep a pointer to predecessor and successor of $v_y$ with color $c$, thus creating a linked list over the leaves for each color $c \in C(r_b)$.

We represent each block $b$ by the root $r_b$ of the tree $T_b$. We keep the roots in a linked list $B$. Since each block has size $\Theta(\log^\beta n)$, $|B| = O( n / \log^\beta n)$. We assign the roots integer labels according to their order using the order maintenance structure of Willard~\cite{Willard1992},
denoted OM. We store the roots of the blocks in a GVEB structure (see Appendix \ref{app:gveb}) $G$ with parameters $(N,M)$  according to their integer labels, where $N = \Theta(|B|)$.  Given a root with integer label $k$ and colors $C(r_b)$, we store one element $(k,c)$ in $G$ for each color $c \in C(r_b)$. OM requires linear space. Insertions and deletions to OM can be performed in $O(\log^2 |B|)$ worst case time, and may cause $O(\log ^2 |B|)$ integer labels to be updated.

\begin{itemize}
\item 
\texttt{FindNext}$(x, C_q)$. This algorithm resembles that of~\cite{Giora2009}. However,
instead of having a bit-vector $C(v)$ for each internal node $v$, $C(v)$ is
represented implicitly by and calculated in constant time from the array
$v.\texttt{leafs}$. We describe the algorithm here for completeness. If $C(x)
\cap C_q \neq \emptyset$, then return $x$. Otherwise $x$ is contained in some
block $b$, and therefore in a tree $T_b$. We traverse the path from $x$ to
$r_b$, and look for the first node $y$ hanging to the right of this path such
that $C_q \cap C(y) \neq \emptyset$. If such a node $y$ exists, then we find the
leftmost leaf in $y$'s subtree with color $c \in C_q$ as follows. Let $v$ be a
the current node we're examining. If $v$ is a leaf, then return it. We repeat
the following until we reach a leaf. Let $v_\ell$ and  $v_r$ be the left and
right children of $v$ respectively. If $C_q \cap C(v_\ell) \neq \emptyset$, then
set $v = v_\ell$. Otherwise set $v = v_r$. Clearly if a there exists a successor
of $x$ in $T_b$ with color $c \in C_q$ we will find it. If not, then we find the
block $b$' with a color in $C_q$ with a query find query on $G$. If $b$' exists
then we repeat the above process to find the leftmost leaf with a color in $C_q$
in $T_b'$. If not, then there is no successor of $x$ with color $c \in C_q$.  It
takes $O(\log \log n)$ time to perform the query in each block, and $O(\log \log
N) = O(\log\log n)$ time to perform the query in $G$. Therefore the total time
required is $O(\log \log n)$.  

\item \texttt{FindPrev}$(x,C_q)$ finds the predecessor of $x$ with color $c \in C_q$ instead
of the successor. It is supported symmetrically.

\item 
\texttt{Report}$(y_1, y_2, C_q)$.  We handle the report query as described by
\cite{Mortensen2003}. This algorithm reports each element $y \in L$ such that
$y_1 \leq y \leq y_2$ once for each color in $C(y) \cap C_q$. Since the number
of colors in $C(y)$ is constant, we can easily modify this procedure to only
report each element once.  Let $T_1$ be the tree for the block containing $y_1$
and let $T_2$ be the tree for the block containing $y_2$. We begin by searching
bottom up in $T_1$ and $T_2$ to report elements to the right of $y_1$ in $T_1$
and to the left of $y_2$ in $T_2$ with colors in $C_q$. We can report these
elements in $O(\log \log n + k)$ time using the \texttt{allleafs} array,
\texttt{leafs} arrays, and the leaf lists. If $T_1 \neq T_2$, then we also
perform a report query in $G$ to find all the blocks between $T_1$ and $T_2$
with colors in $c \in C_q$. For each such block $b$ and color $c$ with root
$r_b$, we can use \texttt{leafs$[c]$} and the leaf list for $c$ to report all
such leaves in constant time per element. 

\item  
\texttt{Mark}$(x,c)$, \texttt{Unmark}$(x,c)$. This is analogous changing the
color of an element $x$ in~\cite{Mortensen2003}. $x$ is stored as a leaf in some
tree $T_b$. We spend $O(\log \log n)$ time updating the nodes on the path from
$x$ to the root $r_b$.  If this adds or removes a color from $C(r_b)$, we make
the corresponding insertion or deletion from $G$ in $O(\log \log n)$ time. We
also need to update the leaf list for color $c$ in $T_b$. Since we can find the
successor and predecessor of $x$ with color $c$ in $O(\log \log n)$ time, we can
make this update in $O(\log \log n)$ time.

\item 
\texttt{Add}$(x)$.
Find the block $b$ where $x$ belongs in $O(\log \log n)$ time, and insert it into $T_b$ in $O(\log\log n)$ time.
Remove an index $i$ from \texttt{b.freelist} and point \texttt{b.allleafs[i]} to $x$ in constant time. Note that insertion may require a block to become too large. However, we can split $b$ into two blocks and amortize the cost of splitting as described in~\cite{Giora2009} or~\cite{Mortensen2003}, so that the amortized cost of an insertion is still $O(\log \log n)$.  

\item 
\texttt{Remove}$(x)$. We do lazy deletion. Instead of removing an element from $T_b$ we mark it as deleted. This requires periodic rebuilding of the entire structure, but does not increase the amortized asymptotic time of the operations.

\end{itemize}
Now we consider the space usage of our structure. Recall that the space required
by a GVEB structure $G$ with parameters $N,M$ is $O(N|C|)$ where $C =
[\log^\frac{1}{4}D]$. We can periodically rebuild $G$ so that $N = \theta(n /
\log^\beta n)$ and $C = O(\log^\frac{1}{4} n)$. Since we set $\beta > 3$, the
total space usage is $O(n)$. The data structure also requires a single lookup
table of size $O(M)$. In our application, $M = O(n)$ as well. 
\end{proof}

\clearpage
\section{One-Dimensional Retroactive Queries}
\label{app:1dRetro}

In this section we build on the GUSF data structure of 
Section~\ref{app:gusf}, and show that we can support one-dimensional fully-retroactive range
reporting and nearest neighbor queries in logarithmic time. 
Note that in one dimension a nearest
neighbor query reduces to a successor and predecessor query.

Let $R \subset \mathbb{R}$ be a set of points we maintain over time. Each point
with value $y$, insertion time $a$ and deletion time $b$ is represented by a
horizontal line segment from $(a,y)$ to $(b,y)$, which we store as the triple
$(a,b,y)$. We construct a modified segment tree as a weight balanced B-tree $T$
with branching parameter $\log^\delta n$ for a small constant $\delta$. The
leaves of the tree represent the endpoints of a segment. A segment $s = (a,b,y)$
is stored at the leaves for its endpoints $a$ and $b$. In addition, we associate
$s$ with all the nodes on the paths from $a$ and $b$ to the root of $T$.  

For each internal node $v$, we maintain a set of segments of associated with $v$
in a GUSF with parameter $N = O(n)$  denoted by $M(v)$.  Each segment $s =
(a,b,y)$ is associated with a color $\overline{c}(s,v) \in C$ that represents
the children of $v$ that are ancestors of $a$ or $b$. We observe that if $v$ has
$\lambda$ children, then we require $O(\lambda^2)$ colors. Therefore we set
$\delta = (1/4) \log ^\frac{1}{8} n $, which limits the number of children so
that the number of colors required is at most $|C| = \log^\frac{1}{4} n$. We
will use the color of segments to guide both queries and updates. We have the
following possible colors. For each distinct pair $u,w$ of children of $v$, we
have a color $c(u,w)$. For each child $u$ of $v$ we also have colors $c(u,u)$
$c_\ell(u)$ and $c_r(u)$. We assign the color $\overline{c}(s,v)$ of segment $s
= (a,b,y)$  as follows. If either $a$ or $b$ are descendants of $v$, then there
must be a child of $v$, $v_a$ or $v_b$ respectively which contain these
endpoints. If both $a$ and $b$ are descendants of $v$,  we let
$\overline{c}(s,v) = c(v_a, v_b)$. Otherwise if only the left endpoint $a$ is a
descendant of $v$, we let $\overline{c}(s,v) = c_\ell(v_a)$. Similarly if only
right endpoint $b$ is a descendant of $v$ we let $\overline{c}(s,v) = c_r(v_b)$.
Thus, if we consider each internal node to represent the range of values of its
leaves, these colors indicate whether $s$ is contained in the range of $v$ or
whether the segment is cut on the left or right by the range. If $s$ either does not
intersect the range of $v$ or spans the range of $v$ 
then neither of its endpoints is a descendant of $v$
and so it is not stored in $M(v)$. 

We maintain the following three tables of~\cite{Giora2009} to facilitate efficient updates and queries on $M(v)$. 
\begin{itemize}
\item 
$U(v)$ is used to update $M(v)$. It maps each pair of children $u,w$ to a color $c(u,w)$. It also maps each child $u$ to the colors $c(u,u)$, $c_\ell(u)$ and $c_r(u)$. 
\item
$Q(v)$ is used to query $M(v)$. It maps each child $u$ to a set of colors
$Q(v)[u]$, all the colors that correspond to segments that span the range of $u$. Thus for each sibling $r$ to the right of $u$ and each sibling $\ell$ to the left of $u$, $Q(v)[u]$ contains the colors $c(\ell, r)$, $c_\ell(\ell)$, and $c_r(r)$. 
\item
$F(v)$ replaces the fractional cascading structure used by~\cite{Giora2009} in their previous implementations. 
It maps each child $u$ of $v$ to the set of colors $F(v)[u] = \{c_\ell(u)\}
\cup \{c_r(u)\} \cup \{c(u,u)\} \cup \{c(u,w) | u \neq w  \textrm{ is a child
of v}\}$, which corresponds to the set of segments for which at least one endpoint is a descendant of $v$. 
\end{itemize}
Let $r$ be the root of our segment tree $T$. Note that $M(r)$ contains all
segments in $T$. We augment $r$ with a balanced search tree $T_r$. We
augment all internal nodes with the tables described above. Given $T$, with the
above augmentations, 
Giora and Kaplan~\cite{Giora2009} proved the space required is only $O(n
\log n / \log \log n)$, and described how to perform insertions, deletions, in
$O(\log n)$ amortized time and fully-retroactive successor queries in
$O(\log n)$ time. We now show how to perform fully-retroactive (one dimensional) range reporting
queries in $O(\log n + k)$ time. 

Given a query report$(t, y, z)$ we can report all the points present at time $t$
with values between $y$ and $z$ as follows. Let $r=v_0, v_1, v_2, \ldots v_t$ be
the search path for $t$ in $T$. First we search for the successor of $y$, $y_0$
and the predecessor of $z$, $z_0$ in $T_r$. 
This gives us pointers to the
smallest and largest elements inside the range, 
(see Figure \ref{fig:AugSegTree}). 
Then for each $i$ in $0 \leq i < k$ we report all the
segments in $M(v_i) \backslash M(v_{i+1})$ by calling 
$M(v_i)$.\texttt{report}$(y, z, Q(v_i)[v_{i+1}])$, 
and we find $y_{i+1}$ and $z_{i+1}$ with two
queries: \texttt{FindNext}$(y_i, F(v_i)[v_{i+1}])$ and  
\texttt{FindPrev}$(z_i, F(v_i)[v_{i+1}])$.
Note that all segments of $M(v_{i+1})$ are contained in $M(v_i)$ and their
colors correspond exactly to the set $F(v_i)[v_{i+1}]$. Hence $y_{i+1}$ and
$z_{i+1}$ are the smallest and largest values respectively of any of the
remaining unreported segments that could possibly intersect the range. Therefore
all segments within the range $[y,z]$ in $M(v_i) \backslash M(v_{i+1})$ will be
reported for each $i$, and correctness follows. To bound the time required, we
observe that we will search twice in $T_r$ and perform one report and two find
queries for each node on the search path to $t$. Therefore we spend $O(\log n)$
time in the root, and $O(\log \log n + k_i)$ time for each of the nodes visited
in $T$. Since the number of nodes on the search path is limited by the height of
$T$ to $O(\log n / \log \log n)$, the total time is $O(\log n + k)$.

\clearpage
\section{Space Filling Curves}
\label{app:space}

Space Filling Curves (SFCs) 
were introduced by Peano~\cite{Peano1890} as a way to
map the unit interval onto the unit square. As a result, two-dimensional SFCs
are
sometimes referred to as Peano curves. SFCs 
are commonly used in computer science
to map multidimensional points down to one dimension.  They have the desirable
property that points that are close to each other along the curve are often also
close to each other in the original multidimensional space.  By now space
filling curves are well-studied and multiple authors have applied them
specifically to quadtrees (e.g., 
\ifFull
see~\cite{Samet1984,Bern1993}) 
\else
see~\cite{Bern1993})
\fi{}
 and approximate
nearest neighbor queries~\cite{Derryberry2008,Chan2006,Chan1998,Liao2001}.
The two most well known SFCs
are the Hilbert Curve introduced by Hilbert~\cite{Hilbert1891}, and the z-curve,  proposed by Morton~\cite{Morton1966}. We call the sorted order of points with respect to their position along the Hilbert or z-curve the Hilbert order or z-order respectively. Some authors (e.g., see~\cite{Faloutsos1989}) have suggested that the Hilbert order achieves better clustering in practice. However, it is more complex to compare the relative positions of two points in the Hilbert order. The z-order has the advantage that comparison of two points only requires bitwise exclusive-or. Thus in this paper we consider the z-order.  

The z-order is sometimes called the shuffle order, because a point's position along the curve can be found by interleaving the bits of its coordinates. The metaphor is that in two dimensions we shuffle together the bits of the $x$ and $y$ coordinates as if they were two halves of a deck of cards.  
Let $P$ be the set of $d$-dimensional input points. 
Assume that each coordinate of each input point $p \in P$ can be represented in binary by a $w$-bit word. 
Let $p_i$ denote the $i$-th coordinate of $p$ and let  $p_{i,j}$ denote the $j$-th bit of $p_i$. Thus $p_i$ is represented as $p_{i,w} \cdots p_{i,1}$ in binary.
Define the \emph{shuffle} $\sigma(p)$ to be the number 
$p_{1,w} \cdots p_{d,w} \cdots p_{1,1} \cdots p_{d,1}$ in binary
Let $P_z = \{\sigma(p) | p \in P\}$, the z-order of $P$. For any two points $p, q \in P$, Chan~\cite{Chan2002} has shown that we can determine their order in $P_z$ in $O(d)$ time as follows:
\begin{quotation} \label{SFC}
\noindent$i$ = 1 \\
for $j$ = 2 to $d$ \\
 \indent   if $|p_i \oplus q_i| < |p_j \oplus q_j|$ \\
\indent\indent i = j\\
return $p_i < q_i$
\end{quotation}
where $\oplus$ denotes bitwise exclusive-or and $|\cdot|$ denotes $\lfloor \log_2 x\rfloor$. However, we can compute 
$|x| < |y|$ without using the $|\cdot|$ operator, by replacing $|x| < |y|$ with this line of code.
\begin{quotation}
\noindent if $x>y$ return false else return $x <  x \oplus y$
\end{quotation}
Thus we can compare $p$ and $q$ using only the $\oplus$ operation.

Although points that are close to each in along an
SFC are close to each other in the original multidimensional space, the
converse it not always true. Points can be close to each other in $\R^d$ but
not close to each other along the curve. Thus in general, points that are
close to each other on the curve may not be nearest neighbors in $\R^d$. In
the rest of this section, we will show that by keeping $d+1$ shifted versions
of the z-order, we can guarantee that for any point $q$, a $c$-approximate
nearest neighbor can be found among the successors and predecessors of $q$ in
one of these z-orders. Here $c$ is a constant that depends on $d$. 

Recall a lemma of Bern {\it et al.}~\cite{Bern1993}, which shows an important relationship between quadtree cells and the z-order. 
\begin{lemma}
\label{shuffleOrderAppendix}
The set of points in any quadtree cell rooted at $[0, 1)^d$ form
a contiguous interval in the z-order.
\end{lemma}
\begin{proof}
For ease of exposition, we give the proof for $d=2$,
but it extends naturally to higher dimensions. 
Let $p \in [0, 1)^2$. We observe that the first two bits of $\sigma(p)$ determine which quadrant of the quadtree contains $p$, as shown in Figure \ref{figQuadrants}. Call this quadrant $Q_1$.  Similarly, the next two bits of $p$ determine which quadrant of $Q_1$ contains $p$. Inductively, if $p$ is in a quadtree cell $Q_{i}$ determined by the first $2i$ bits of $p$, then the next two bits of $p$ determine which quadrant of $Q_{i}$ contains $p$.  Thus, there is a bijection between a path in the quadtree from the root to a cell on level $i$ and the first $2i$ bits of $p$. All points within the same quadtree cell on level $i$ must have the first $2i$ bits in common. 
Suppose $Q_i$ is a quadtree cell, and that $I$ is the set of points in the
z-order contained in $Q_i$. As we have shown, these points must share the
first $2i$ bits. Now suppose we have a point $q \notin Q_i$. Then one of the
first $2i$ bits of $\sigma(q)$ must differ from the first $2i$ bits of $I$.
Therefore, $\sigma(q)$ must either come before $I$ or after $I$ in the
z-order. It cannot occur between two elements of I. Otherwise $\sigma(q)$
would not have had a bit that differs from the first $2i$ bits in $I$, which would imply that $q \in Q_i$. 
\end{proof}

We say that two points $p,q$ belong to the same $r$-grid cell iff $p \mydiv r = q \mydiv r$.
 We say that a point $p$ is $\alpha$-central in its $r$-grid cell iff for each integer $i \in [1,d]$, we have $\alpha r \leq p_i \mod r < (1-\alpha)r$, or equivalently $(p_i + \alpha r) \mod r \geq 2\alpha r$. 
The following lemma is due to Chan~\cite{Chan1998}.
\begin{lemma}
\label{rCentral}
Suppose $d$ is even. Let $v^{(j)} = (j/(d+1), \ldots, j/(d+1)) \in \R^d.$ 
For any point $p \in \R^d$ and $r = 2^{-\ell} (\ell \in \N)$, there exists $j \in \{0,1, \ldots d\}$ such that $p + v^{j}$ is $(1/(2d+2))$-central in its $r$-grid cell.  
\end{lemma}

The next lemma was proven by~\cite{Liao2001} for the Hilbert curve, and~\cite{Lopez2000}
for the z-curve. 
\begin{lemma}
Let $C_\epsilon(q)$ denote the closed hypercube with edges of length $2\epsilon$. centered on $q$.
 Given a point $q \in [0, 1)^d$ and $\epsilon$, where $0 < \epsilon < 1/(2d+2)$, there exists a $j, 0 \leq j \leq d$ such that $C_\epsilon(q + v^{j})$ is contained in an $r$-region with 
 $r$ satisfying $r/(4d+4) \leq \epsilon < r/ (2d+2)$.
\end{lemma}

Given the previous lemmas,
Liao {\it et al.}~\cite{Liao2001} showed the following theorem. We leverage this result to achieve fully-retroactive $(1+\epsilon)$-approximate nearest neighbor queries.
\begin{theorem}
Let $P$ be a set of points in $\R^d$. Define a constant $c = \sqrt{d}(4d+4) + 1$. 
Suppose that we have $d+1$ lists $P + v^{j}, j= 0, \ldots, d$, each one sorted according to its z-order.  We can find a query point $q$'s $c$-approximate nearest neighbor in $P$ by examining the the $2(d+1)$ predecessors and successors of $q$ in the lists. 
\end{theorem}
\begin{proof} 
Let $\delta$ be the distance between a query point $q$ and its true nearest neighbor $p \in P$. 
By Lemma \ref{shuffleOrderAppendix}, the $r$-region $Q_j$ that contains both $q$ and $p$ contains all points in the z-order between $q$ and $p$. Thus $Q_j$ must contain the predecessor or the successor of $q$. 
Furthermore, by Lemma \ref{rCentral},  there exists an $r$-region that contains the hypercube 
$C_\delta(q + v^{(i)})$  such that $r/(4d+4) \leq \delta < r/(2d+2)$, for some $0 \leq i \leq d$. 
Therefore one of the $2(d+1)$ successors and predecessors is at most distance $r \leq \sqrt{d}(4d+4)\delta$  to $q$. 
\end{proof}

\label{app:figs}
\begin{figure}
\centering
\includegraphics[scale=1]{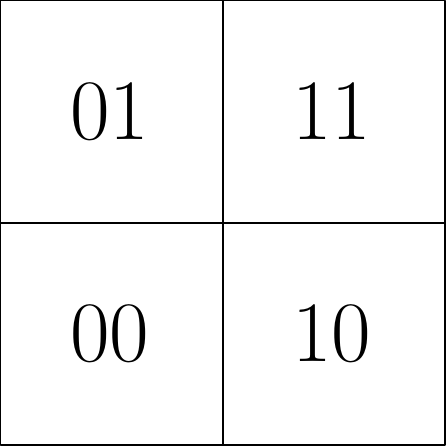}
\caption{\label{figQuadrants} Suppose $p \in [0,1)^2$. The next two bits in the shuffle of $p$ correspond to which quadrant of the cell contains $p$.}
\end{figure}
\begin{figure}
\centering
\includegraphics[width=.95\textwidth]{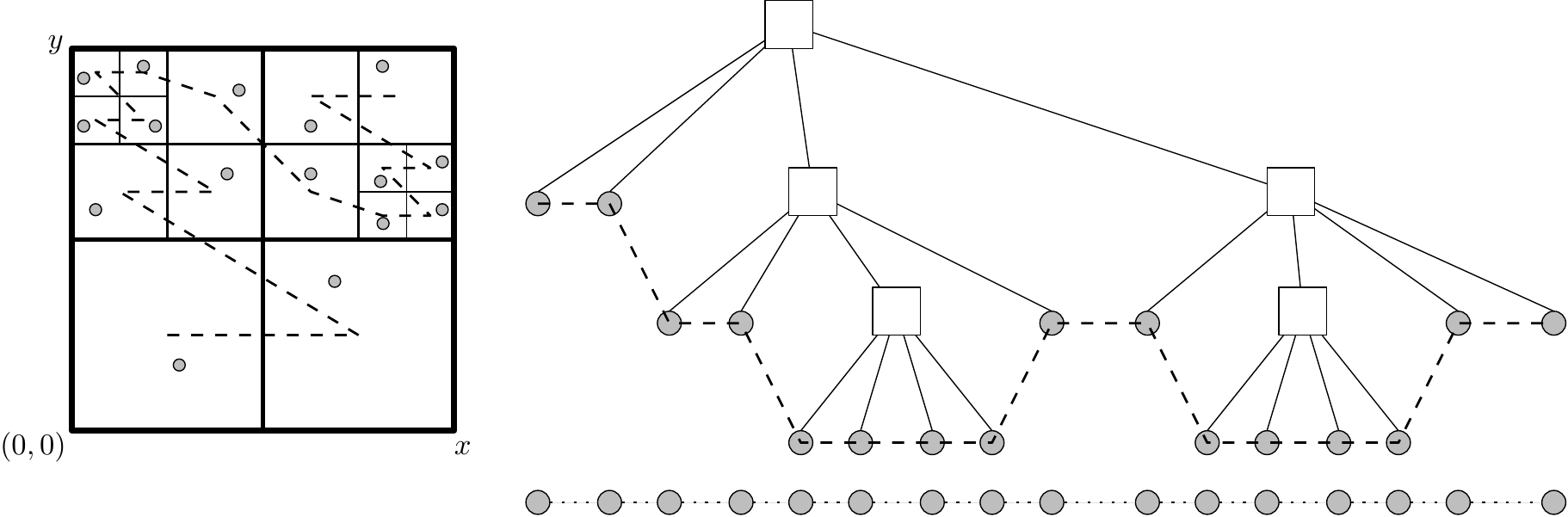}
\caption{
\label{fig:zorder}
The z-order curve corresponds to the order the leaves would be visited in an in-order traversal of the quadtree.
We can store the points in a linked list in this order. Any element between two elements $i,j$ in the linked list must fall in the same quadtree cell as $i$ and $j$. }
\end{figure}

\end{appendix}

\end{document}